\documentclass[10pt,a4paper,twoside]{amsart}
\usepackage{amsfonts, amsmath, amssymb,latexsym}
\usepackage[english]{babel}
\usepackage[utf8]{inputenc}
\usepackage{hyperref}

\newtheorem{thm}{Theorem}[section]
\newtheorem{lem}[thm]{Lemma}
\newtheorem{prop}[thm]{Proposition}

\theoremstyle{definition}
\newtheorem{defn}[thm]{Definition}
\theoremstyle{remark}
\newtheorem{rem}[thm]{Remark}
\theoremstyle{examples}
\newtheorem{ex}[thm]{Example}

\numberwithin{equation}{section}

\newcommand{\Tr}{\operatorname{Tr}}

\title[Trace-based cryptanalysis of PLWE\ldots]%
{Cryptanalysis of PLWE based on\\zero-trace quadratic roots}

\author[B. Barbero]{Beatriz Barbero-Lucas}
\author[I. Blanco]{Iv\'an Blanco-Chac\'on}
\author[R. Dur\'an]{Ra\'ul Dur\'an-D\'{\i}az}
\author[R. Mart\'{\i}n]{Rodrigo Mart\'{\i}n S\'anchez-Ledesma}
\author[R. Y. N. Nchiwo]{Rahinatou Yuh Njah Nchiwo}
\address{School of Mathematics and Statistics, University College Dublin,
Ireland;
Department of Physics and Mathematics, University of Alcal\'a, Spain;
Polytechnic School, University of Alcal\'a;
Departamento de \'Algebra, Universidad Complutense de Madrid, Spain, and
Indra Sistemas de Comunicaciones Seguras, Spain;
Aalto University School of Science, Finland}
\email{beatriz.barberolucas@ucdconnect.ie; ivan.blancoc@uah.es;
raul.duran@uah.es; rodrma01@ucm.es,
rmsanchezledesma@indra.es; rahinatou.njah@aalto.fi}
\thanks{B. Barbero-Lucas is partially supported by the University of Alcal\'a
grant CCG20/IA-057. I. Blanco-Chac\'on is partially supported by the Spanish
National Research Plan, grant n. MTM2016-79400-P, by grant PID2019-104855RBI00,
funded by MCIN / AEI / 10.13039 / 501100011033 and by the University of
Alcal\'a grant CCG20/IA-057. R. Dur\'an-D\'iaz is partially supported by grant
P2QProMeTe (PID2020-112586RB-I00), funded by MCIN / AEI / 10.13039 /
501100011033. R.Y. Njah Nchiwo is supported by a PhD scholarship from the Magnus
Ehrnrooth Foundation, Finland, in part by Academy of Finland, grant 351271 (PI:
C. Hollanti) and in part by MATINE Finnish Ministry of Defence, grant
\#2500M0147 (PI: C. Hollanti).}

\begin{document}
\def\figurename{Algorithm}
\pagestyle{myheadings} \markboth{\evenhead}{\oddhead}
\thispagestyle{empty}

\begin{abstract}We extend two of the attacks on the PLWE problem presented in
\cite{ELOS:2016:RCN} to a ring $R_q=\mathbb{F}_q[x]/(f(x))$ where the
irreducible monic polynomial $f(x)\in\mathbb{Z}[x]$ has an irreducible quadratic
factor over
$\mathbb{F}_q[x]$ of the form $x^2+\rho$ with $\rho$ of suitable multiplicative
order in $\mathbb{F}_q$. Our attack exploits the fact that the trace of the root
is zero and has overwhelming success probability as a function of the number of
samples taken as input. An implementation in Maple and some examples of our
attack are also provided.
\end{abstract}

\maketitle

\section{Introduction}

Lattice-based cryptography has achieved an extraordinary success in the last
years, especially since the National Institute for Standards and Technology
launched a public contest in 2017 to standardise different sets of
quantum-resistant primitives. In 2022, four candidates were proposed for standardisation,
three of them belonging to the lattice-based category (Crystals-Kyber for
Encryption and Key Encapsulation, Crystals-Dilithium and Falcon for Digital
Signatures). These suites are based on several \emph{learning problems} over
certain large-dimensional vector spaces over finite fields, some of them with an
extra ring structure.

In the present work, we are concerned in particular with the Polynomial Learning
With Errors problem (PLWE from now on), which was introduced in
\cite{SSTX:2009:EPK}. To briefly describe it, let us denote by $R$ the quotient
ring $R:=\mathbb{Z}[x]/(f(x))$ with $f(x)\in\mathbb{Z}[x]$ of degree $N$, monic
and irreducible over $\mathbb{Z}[x]$, and let $q\geq 2$ be a prime. Set $R_q:=
\mathbb{F}_q[x]/(f(x))$ and consider an $R_q$-valued uniform random distribution
$U$ and an $R_q$-valued discretized Gaussian distribution $\chi_{R_q}$. The PLWE
problem in decision form, the version we investigate in this work, consists in
guessing, with non-negligible advantage, if a set of pairs $\{(a_i(x),b_i(x))\}
_{i\geq 1}\subseteq R_q\times R_q$ of arbitrary size has been sampled either
from $U\times U$ or from the PLWE distribution (see Definition~\ref{plwedistr}
for details).

In nearly all practical instances of (P)LWE-based cryptographic schemes, either in post-quantum cryptography or in
homomorphic encryption, the polynomial $f(x)$ defining  $R_q$ is considered to
be cyclotomic, due to the distinctive arithmetic properties enjoyed by these
fields, which are used in the security reduction proofs as well as in practical
implementations. For instance, the NIST standard Crystals-Kyber
(\cite{FIPS:2024:203}) is based on the Module Learning With Errors problem, a variant
of PLWE. In particular, its modulus polynomial is the $2^8$ cyclotomic polynomial. The third-round NIST candidate SABER
(\cite{saber}) also uses a variant of PLWE backed on the $2^8$-th cyclotomic polynomial. As
for homomorphic encryption standards, nearly all the schemes use PLWE moduli
polynomials that are either cyclotomic (mainly HELib \cite{helib}, LATTIGO
\cite{lattigo}, or OpenFHE \cite{openfhe}) or minimal polynomials of the maximal
totally real extension (the current implementation of LATTIGO allows to
choose this subfield). For a complete analysis of homomorphic encryption schemes
based on non-cyclotomic number fields and their practical performance we refer
the reader to \cite{PTGG:2021:RMR}.

However, some other works like \cite{PP:2019:ASL} and
\cite{RSW:2018:RPP} deal with more general families of non-cyclotomic number
fields and polynomials. Moreover, the problem is shown to enjoy a reduction from a supposedly
\emph{hard} problem for the family of ideal lattices whenever the prime $q$ is
of the form $q=\mathcal{O}(N^k)$ with $k\in\mathbb{N}$ independent of $N$ (for
instance, see Lemma~4.1 in \cite{SSTX:2009:EPK} for the power-of-two cyclotomic case).
As a matter of fact, in most applications $k$ is considered to be $2$ and this is
what we will also assume here.

The aforementioned ideal-lattice based problems are supposed to be intractable in
the worst case although some further investigation shows that the Shortest Vector
Problem might be more tractable for ideal lattices than for general ones (for
an up-to-date account on the ideal-lattice Shortest Vector Problem in the cyclotomic
case see \cite{DPW:2019:SVF}, while for more general families of ideal lattices
the reader is referred to \cite{BR:2020:TPU}).

Beyond worst-case reductions, some particular instantiations of PLWE have been
encountered to be insecure. For instance, in \cite{ELOS:2015:PWI}
and \cite{ELOS:2016:RCN}, the authors provide an attack against PLWE whenever a
simple root $\alpha\in\mathbb{F}_q$ of the polynomial $f(x)$ exists such that
either
\begin{enumerate}
\item[\emph{a)}] $\alpha=1$, or
\item[\emph{b)}] $\alpha$ has \emph{small order} modulo $q$, or
\item[\emph{c)}] $\alpha$ has \emph{small residue} modulo $q$, or
\item[\emph{d)}] none of the above situations holds but $\alpha$ and the variance of the error distribution are
such that one can distinguish the PLWE distribution from the uniform one.
\end{enumerate}
If $f(x)=\Phi_n(x)$, the $n$-th cyclotomic polynomial, it is easy to check that
$\alpha=\pm 1$ is never a root of $\Phi_n(x)$ modulo $q$ (unless $q\mid n$) and
all the roots of $\Phi_n(x)$ have maximal order. For instance, if $q\equiv
1\pmod{n}$ so that $q$ is totally split over the ring of integers of the
cyclotomic field, this order is precisely $n$. Hence any attack based on
conditions \emph{a)} or \emph{b)} does not apply to cyclotomic rings. It is however possible
to attack in polynomial time the PLWE problem restricted to a large enough
subring whenever a prime-power-degree cyclotomic polynomial splits in nonlinear
factors of suitable small degree (see \cite{BDNB:2023:TBC}). This earlier attack
addresses a situation not dealt with in \cite{ELOS:2015:PWI}, which assumes that
the modulus polynomial is totally split, but the assumption that the samples
must belong to a smaller subring imposes a strong limitation on the attack to
work.

The family of attacks based on case \emph{d)} has been extended in
\cite{BDM:2025:GARB}, \cite{BDM:GAR:2025} by the second and third author, jointly with R. Martín Sánchez-Ledesma. In that recent work, new sets of dangerous instances have been identified. In
addition, \cite{BDM:2025:GARB} elaborates on the results of the present work
and provides new attacks when the modulus polynomial admits a cubic or quartic
irreducible factor, whose roots have zero trace and a norm with \emph{small}
order. However, even if \cite{BDM:2025:GARB} has been published first, the present work provides foundational background which is omitted and taken for granted in \cite{BDM:2025:GARB}. In particular, in the present work we prove Proposition 3.5, which provide the complexity computation of our main algorithm. Likewise, we prove Lemma 3.6 and Lemma 3.7,  which reduce the analysis of traces of errors evaluated on roots in $\mathbb{F}_{q^2}$ to the identification of a smallness region in $\mathbb{F}_q$.

Further results on weak
instances for the PLWE problem can be found in \cite{Peikert:2016:HIR}, exploiting
algebraic and arithmetic vulnerabilities, and in
\cite{CLS:2017:ASR},\cite{Stange:2021:AAS}, exploiting a more statistical approach
(the so-called Chi-square attack) which generalises the attack for the case
\emph{c)} above.

The main contribution of the present work is to offer an extension of the attack
developed by \cite{ELOS:2015:PWI} to the case in which $f(x)$ has a zero-trace
root $\alpha\in\mathbb{F}_{q^2}\setminus\mathbb{F}_q$ such that the norm
$N_{\mathbb{F}_{q^2}\mid\mathbb{F}_q}(\alpha)$ has a \emph{suitable}
multiplicative order in relation with the rest of parameters in such a way that
the trace of the error polynomials evaluated at $\alpha$ belong to a region
whose cardinality is small enough in relation with $q$. This condition is far
less restricting than requiring the order of the norm to be very small.

Our attack proceeds by recursively evaluating an input set of samples in
$R_q\times R_q$ at the root $\alpha$. After each evaluation, we take the trace
of the result and check, as in
\cite{ELOS:2016:RCN}, whether this trace belongs to a certain distinguished
region, whose cardinal is much smaller than $q$.

The general strategy of attacking the decision PLWE problem via reduction and/or trace
maps to smaller spaces is already well known. For instance, in
\cite[\S 3.2.1]{Peikert:2016:HIR} it is discussed how the dual trace pairing can be used to
detect whether the reduction of the error distribution modulo a prime ideal of
$R$ over $q$ is uniform or not; but our approach is totally different,
since we do not resort to dual lattices at all.

We have structured our paper in four sections: Section 2 recalls several
definitions and notations, as well as some facts on random variables over
finite fields, which we will need later on. Section 3 consists of two
subsections: the first one recalls in detail the attack developed by
\cite{ELOS:2015:PWI}; motivated by the latter, the second one introduces our
attack in two stages: the
first stage is an attack against samples whose uniform first component belongs
to a certain distinguished subring $R_{q,0}\subseteq R_q$ and the second stage
is a general attack which reduces to the previous one in probabilistic
polynomial time on the size of the set of samples and the parameters of the
quotient ring. Finally, Section 4 describes simulations of our attack using Maple.
The Maple code for our examples has been uploaded to a \textsc{GitHub} domain
referenced therein.

Finally, we must stress out that neither this work nor \cite{ELOS:2015:PWI}, \cite{EHL:2014:WIP}, \cite{Peikert:2016:HIR}, \cite{RSW:2018:RPP} does not attempt to recommend or rule out any precise sets of parameters for PLWE-based cryptosystems. Our aim is just to provide a further study of algebraic cryptanalysis towards these schemes, which is to be regarded as a natural continuation of the aforementioned works and shows how the trade-off between the order of the traces modulo $q$, the variance and the dimension impact on the security of the cryptosystem.

\section{Preliminary facts and notations}
Before recalling the definition of the PLWE problem, we find it convenient to
explicitly mention and prove a few facts about discrete uniform and
discrete Gaussian random variables over finite fields, as they are at the very
background of the learning problems under study and, moreover, some of our
arguments in Section 3 will involve certain manipulations of uniform
distributions supported over finite dimensional vector spaces over
$\mathbb{F}_q$.

\begin{defn}Let $E$ be an $\mathbb{F}_q$-vector space of dimension $d$ so that
$|E|=q^d$. A random variable $X$ with values over $E$ is said to be uniform if
for every $v\in E$ we have that $P[X=v]=q^{-d}$.
\end{defn}
We will use the following fact:
\begin{lem}(\cite[Lemma 2.2]{BDNB:2023:TBC})
If $X_1,X_2,\dotsc,X_N$ are independent uniform distributions over
$\mathbb{F}_q$ then, for each
$\lambda_1,\lambda_2,\dotsc,\lambda_N\in\mathbb{F}_q$, not all of them zero, the
variable $\sum_{i=1}^N\lambda_i X_i$ is uniform.
\label{unifl}
\end{lem}
We will also make use of the following fact:
\begin{prop}
Let $X$ be a uniform distribution on the vector space $\mathbb{F}_q^N$
and let $A:\mathbb{F}_q^N\to\mathbb{F}_q^M$ be a linear map. Then, the distribution
$AX$ obtained by \emph{a)} sampling an element $\mathbf{x}\in\mathbb{F}_q^N$ uniformly
from $X$ and \emph{b)} returning $A\mathbf{x}$, is also uniform over $\mathrm{Im}(A)$. In
particular, the $i$-th component is uniformly sampled from $\mathbb{F}_q$.
\label{unifprod}
\end{prop}
\begin{proof}
The result is clear taking into account that for each $\mathbf{y}\in\mathrm{Im}(A)$
the number of elements in the preimage $A^{-1}(\mathbf{y})$ is constant and
independent from $\mathbf{y}$ (actually there are $|\mathrm{Ker}(A)|$, each of
them occurring with probability $1/q^N$).
\end{proof}

The second statistical component of our learning problem is the notion of a
discrete Gaussian random variable. Here we follow closely \cite[Section
2]{ELOS:2015:PWI} and \cite[Section 2]{LPR:2013:ILL}.

As it is well known, for $r>0$, the Gaussian function
$\rho_r:\mathbb{R}^N\to\mathbb{R}_{>0}$ defined by
$\rho_r(\mathbf{x})=\exp(-\pi||\mathbf{x}||^2/r^2)$ yields, after normalising,
the density function of a continuous Gaussian random variable of zero mean and covariance matrix $r^2 I_N$, where $I_N$ is the $N$-dimensional
identity matrix. Furthermore, let $\{\mathbf{h}_i\}_{i=1}^N$ be a basis of
$\mathbb{R}^N$.  Given a vector $\mathbf{r}=(r_1,\dotsc,r_N)\in\mathbb{R}_{+}^N$,
suppose that, for each $r_i$, $D_{r_i}$ is a $1$-dimensional continuous Gaussian
random variable of
zero mean and variance $r_i$. If these $N$ random variables are independent,
then the variable $D_{\mathbf{r}}=\sum_{i=1}^ND_{r_i}\mathbf{h}_i$ is a
continuous elliptic $N$-dimensional Gaussian variable of zero mean. Its covariance matrix has $\mathbf{r}$ as main diagonal and $0$ elsewhere.
Denote by $\rho_{\mathbf{r}}(\mathbf{x})$ the density function of
$D_{\mathbf{r}}$. In the particular case that all the components of the
vector $\mathbf{r}$ are equal, then we say that the Gaussian variable is
spherical.

Now, suppose that $\mathcal{L}$ is a full-rank lattice contained in
$\mathbb{R}^N$ (namely, an Abelian additive subgroup of $\mathbb{R}^N$ finitely
generated over $\mathbb{Z}$  of rank $N$) and suppose that
$\{\mathbf{h}_i\}_{i=1}^N$ is a
$\mathbb{Z}$-basis of $\mathcal{L}$, hence a basis of $\mathbb{R}^N$ as a
vector space.

\begin{defn}A discrete random variable $X$ supported on $\mathcal{L}$ is called
a discrete elliptic Gaussian random variable if its probability function is
given by
$$
P[X=\mathbf{x}]=\frac{\rho_{\mathbf{r}}(\mathbf{x})}
                     {\rho_{\mathbf{r}}(\mathcal{L})}
\quad
\text{for any }\mathbf{x}\in\mathcal{L}\text{ where }\rho_{\mathbf{r}}(\mathcal{L})=\sum_{\mathbf{x}\in\mathcal{L}}\rho_{\mathbf{r}}(\mathbf{x}).
$$
\end{defn}

\subsection{The PLWE problem}
Suppose that $f(x)\in\mathbb{Z}[x]$ is a monic irreducible polynomial of degree
$N$. Denote, as in the introduction, $R:=\mathbb{Z}[x]/(f(x))$ and for a prime
$q \in\mathbb{Z}$, let us set $R_q:=R/qR\cong\mathbb{F}_q[x]/(f(x))$.

We can naturally endow the ring $R$ with a lattice structure over
$\mathbb{R}^N$ via the following group monomorphism known in the literature as the coefficient embedding:
\[
\begin{array}{ccc}
\sigma_{coef}: R & \hookrightarrow & \mathbb{R}^N\\
\sum_{i=0}^{N-1}a_i\overline{x}^i & \mapsto & (a_0,\dotsc,a_{N-1}),
\end{array}
\]
where $\overline{x}^i$ stands for the class of $x^i$ modulo the ideal $(f(x))$.

The lattice $\sigma_{coef}(R)$ inherits an extra multiplicative structure from
$R$ and deserves a special terminology in the lattice-based cryptography
literature:
\begin{defn}An ideal lattice is a lattice $\mathcal{L}$ such that there exists a
ring $R$, an ideal $I\subseteq R$, and an additive group monomorphism $\sigma:
R\hookrightarrow\mathbb{R}^N$ such that $\mathcal{L}=\sigma(I)$.
\end{defn}

For a prime $q\geq 2$, if $\chi_R$ is a discrete Gaussian random variable
with values in $R$ (or, rather, in its image $\sigma_{coef}(R)\subseteq
\mathbb{R}^N$ by the coefficient embedding) we can also reduce its outputs modulo
$q$. This is hence a random variable with finite support and we refer to it as a
discrete Gaussian variable modulo $q$ and denote it by $\chi_{R_q}$.

Now, let $q \in \mathbb{Z}$ be a fixed prime and let $\chi_{R_q}$ be such a
$q$-reduced discrete Gaussian random variable of $0$ mean, hence with values in
$R_q$. Then we have the following

\begin{defn}[The PLWE distribution]
For a fixed element $s\in R_q$ and an error distribution $\chi_{R_q}$ over $R_q$
the PLWE distribution attached to the triple $(R_q,s, \chi_{R_q})$ is the
probabilistic algorithm $\mathcal{A}_{s, \chi_{R_q}}$ that works as follows:
\begin{enumerate}
\item Samples an element $a \in  R_q$ from a uniform random sampler.
\item Samples an element $e\in R_q$ from $\chi_{R_q}$.
\item Returns the element $(a, b=as+e)$.
\end{enumerate}
\label{plwedistr}
\end{defn}

Denoting by $R^2_q$ the Cartesian product $R_q\times R_q$, the PLWE problem for
the triple $(R,q,\chi_{R_q})$ is defined as
follows. We give the decision version, since that is all we need for our
purpose.

\begin{defn}[PLWE decisional problem]Let $\chi_{R_q}$ be defined as before. Let $\mathcal{A}$ be the PLWE distribution attached to a triple $(R_q,s, \chi_{R_q})$ for some $s\in R_q$ and $\mathcal{B}$ the uniform distribution on $R_q^2$. The (decisional) PLWE problem consists in deciding with non-negligible advantage, for a set of
samples of arbitrary size  $(a_i, b_i)\in R_q^2$, whether they are sampled from $\mathcal{A}$ or from $\mathcal{B}$.
\label{plwedefn}
\end{defn}

\section{\texorpdfstring{An attack based on traces over finite extensions of
$\mathbb{F}_q$}{An attack based on traces over finite extensions of Fq}}

The first part of this Section will present in detail the attack developed in
\cite{ELOS:2015:PWI}, since our proposal is an extension of such attack to the
case in which $f(x)$ has a zero-trace root $\alpha\in\mathbb{F}_{q^2}\setminus
\mathbb{F}_q$. Our attack will be then described in the second part of the
section.

We point out that although not explicitly stated in Definition \ref{plwedefn} of
the PLWE problem, it is assumed that the error distribution is a zero-mean discrete
spherical $q$-reduced Gaussian, $\chi_{R_q}$. Remark that most
security reduction proofs actually assume that $\chi_{R_q}$ is spherical: this
assumption is legitimate as long as only a uniform upper-bound on the diagonal
entries is required, as is the case in most of these security reduction proofs
and also in ours.

Moreover, we will assume, as in \cite{ELOS:2015:PWI}, that $\chi_{R_q}$ is
truncated at $2\sigma$, namely, that each of the independent Gaussian
distributions corresponding to each coordinate is so truncated and renormalized. In particular, the support  of each coordinate's probability mass function -is contained in $[-2\sigma, 2\sigma]\cap\mathbb{Z}$. 

\subsection{\texorpdfstring{Attack developed by \cite{ELOS:2015:PWI} and
\cite{ELOS:2016:RCN}}{Attack developed by ELOS:2015:PWI and ELOS:2016:RCN}}

In those references, the authors describe an attack against the PLWE
problem attached to a quotient ring $R_q=\mathbb{F}_q[x]/(f(x))$ and a prime
$q$, if there exists a simple root $\alpha\in\mathbb{F}_q$ such that either
\begin{enumerate}
\item[\emph{a)}] $\alpha=1$, or
\item[\emph{b)}] $\alpha$ has \emph{small order} modulo $q$, or
\item[\emph{c)}] $\alpha$ has \emph{small residue} modulo $q$, or
\item[\emph{d)}] none of the above situations holds but $\alpha$ and the variance are
such that one can distinguish the PLWE distribution from the uniform one.
\end{enumerate}
The expressions \emph{small order} and \emph{small residue}
are not made fully precise therein, but the examples provided by the authors show
that by such a term we must understand orders of up to $5$ while by \emph{small
residue} the authors understand $\alpha=2,3$. We will not address cases \emph{c)} and
\emph{d)}, which are investigated in \cite{BDM:2025:GARB}.

By the Chinese remainder theorem, we can write
\[
R_q\simeq \mathbb{F}_q[x]/(x-\alpha)\times\mathbb{F}_q[x]/(h(x)),
\]
with $h(x)$ coprime to $x-\alpha$. Then we have the ring homomorphism
\[
\psi_\alpha\colon R_q\rightarrow\mathbb{F}_q[x]/(x-\alpha)\simeq\mathbb{F}_q,
\]
which is precisely the evaluation map, namely, $\psi_\alpha(g(x)) = g(\alpha)$,
for any $g(x)\in R_q$.

Assume that $\alpha=1$ is a root of $f(x)$.
Given a
PLWE sample $(a(x),b(x)=a(x)s(x)+e(x))$, the error term,
$e(x)=\sum_{i=0}^{N-1}e_ix^i$, has its coefficients
$e_i\in\mathbb{F}_q$ sampled from $\chi_{R_q}$ of small
enough standard deviation $\sigma$ (the authors consider ranges for $\sigma$ in
the interval $[0.3, 8]$, as
suggested by the applications they analyse). Applying the evaluation map, we have
\[
b(1)-a(1)s(1)=e(1)=\sum_{i=0}^{N-1}e_i,
\]
and the sum $\sum_{i=0}^{N-1}e_i$ hence takes values on the set $[-2N\sigma,2N\sigma]\cap\mathbb{Z}$.

If one loops over $\mathbb{F}_q$ for a guess for $s(1)$, a right guess
$g$ is one such that the value $b(1)-a(1)g$ belongs to $\Sigma$. Since the error distribution is assumed to be a centered Gaussian truncated at $2\sigma$, we can detect which distribution do the samples come from (uniform or PLWE) with significant probability. We will refer to the set $\Sigma$ as the \emph{smallness region} for this attack.

In the same paper, the authors develop a more subtle attack, based on the
multiplicative order of the roots of $f(x)$. Suppose that $r$ is the order of a
root $\alpha\in\mathbb{F}_q, \alpha\ne 1$. Again given a PLWE sample
$(a(x),b(x)=a(x)s(x)+e(x))$, we have
\[
b(\alpha)-a(\alpha)s(\alpha)=e(\alpha) =
\sum_{j=0}^{r-1}\sum_{i=0}^{\left\lfloor\frac{N}{r}\right\rfloor-1}
e_{ir+j}\alpha^j.
\]

Now each coefficient $\epsilon _j := \sum_{i=0}^{\left\lfloor\frac{N}{r}\right\rfloor-1}
e_{ir+j}$ takes values on $[-2\frac{N}{r}\sigma, 2\frac{N}{r}\sigma]\cap\mathbb{Z}$. Hence,  in this case we define the \emph{smallness region}:
$$
\Sigma:=\left\{\sum_{j=0}^{r-1}x_j\alpha^j:\,\,x_j\in\left[-2\frac{N}{r}\sigma, 2\frac{N}{r}\sigma\right]\cap\mathbb{Z}\right\},
$$
namely, the set of all possibilities for $e(\alpha)$. $\Sigma$ can be precomputed and stored as a look-up table.

Observe that
\[
|\Sigma| \leq \left(4\frac{N}{r}\sigma + 1\right)^r.
\]

\begin{figure}[ht]
\centering
\begin{tabular}[c]{ll}
\hline
\textbf{Input:} & A collection of samples $C=\{(a_i(x),b_i(x))\}
_{i=1}^M\subseteq R_q^2$ \\
& A look-up table $\Sigma$ of all possible values for $e(\alpha)$ \\
\textbf{Output:} & A guess $g\in\mathbb{F}_q$ for $s(\alpha)$,\\
 & or \textbf{NOT PLWE},\\
 & or \textbf{NOT ENOUGH SAMPLES}\\
\hline
\end{tabular}

\medskip

\begin{itemize}
\item \texttt{\emph{set}} $G:=\emptyset$
\item \texttt{\emph{for}} $g\in \mathbb{F}_q$ \texttt{\emph{do}}
	\begin{itemize}
	\item \texttt{\emph{for}} $(a_i(x),b_i(x))\in C$  \texttt{\emph{do}}
		\begin{itemize}
		\item \texttt{\emph{if}} $b_i(\alpha)-a_i(\alpha)g\notin\Sigma$ \texttt{\emph{then}}
			\begin{itemize}
			\item \texttt{\emph{next}} $g$
			\end{itemize}
		\end{itemize}
	\item \texttt{\emph{set}} $G:=G\cup\{g\}$
	\end{itemize}
\item \texttt{\emph{if}} $G=\emptyset$ \texttt{\emph{then return}} \textbf{NOT PLWE}
\item \texttt{\emph{if}} $G=\{g\}$ \texttt{\emph{then return}} $g$
\item \texttt{\emph{if}} $|G|>1$ \texttt{\emph{then return}} \textbf{NOT ENOUGH SAMPLES}
\end{itemize}
\hrule
\caption{Algorithm solving PLWE decision problem}
\label{alg1}
\end{figure}
These ideas turn in \cite{ELOS:2015:PWI} into the Algorithm~\ref{alg1}, whose
probability of success is given by the following result:

\begin{prop} Let $\Sigma$, $M$ and $r$ be defined as above, and assume $|\Sigma|<q$. Then
\begin{enumerate}
\item
If Algorithm~\ref{alg1} returns \textbf{PLWE} or \textbf{NOT ENOUGH SAMPLES},
then the samples are \textbf{PLWE} with probability at least $1-q\left(\frac
{|\Sigma|}{q}\right)^M$.
\item
If Algorithm~\ref{alg1} returns \textbf{NOT PLWE}, then the samples are
uniform with probability $1$.
\end{enumerate}
\label{pralg1}
\end{prop}
\begin{proof}
\begin{enumerate}
\item
Let $E$ represent the event ``Algorithm~\ref{alg1} returns \textbf{PLWE} or
\textbf{NOT ENOUGH SAMPLES}''. By using Bayes' theorem we have
\begin{eqnarray}\label{eq.1}
&& \mathrm{P}(D=\mathcal{G}_\sigma | E) = 1 - \mathrm{P}(D=U|E) \nonumber \\
&& = 1 -
\frac{\mathrm{P}(E|D=U)\cdot\mathrm{P}(D=U)}
{\mathrm{P}(E|D=U)\cdot\mathrm{P}(D=U)+
\mathrm{P}(E|D=\mathcal{G}_\sigma)\cdot\mathrm{P}(D=\mathcal{G}_\sigma)}
\nonumber \\
&& = 1 - \frac{\mathrm{P}(E|D=U)}{\mathrm{P}(E|D=U) +
\mathrm{P}(E|D=\mathcal{G}_\sigma)}.
\end{eqnarray}
Notice that we are implicitly assuming that both distributions can be chosen with equal probability. 

Observe that $\mathrm{P}(E|D=\mathcal{G}_\sigma)$ is the probability that there exist some $g\in\mathbb{F}_q$ such that $b_i(\alpha)-a_i(\alpha)g\in
\Sigma$, for all $i=1,\dotsc,M$. Since we are supposing that the distribution is precisely $\mathcal{G}_\sigma$, then at least one $g$ is sure to exist, namely, $s(\alpha)$, the evaluation of the secret. Then we have that
$\mathrm{P}(E|D=\mathcal{G}_\sigma)\geq \mathrm{P}(s(\alpha)\in\Sigma)=1$.

To compute $\mathrm{P}(E|D=U)$, denote by $U_g$ the event that $b_i(\alpha)-a_i
(\alpha)g\in\Sigma$, for all $i=1,\dotsc,M$. Then
\[
\mathrm{P}(E|D=U) = \mathrm{P}\left(\bigcup_{g=0}^{q-1} U_g\right) \leq
\sum_{g=0}^{q-1}\mathrm{P}(U_g),
\]
as the events $U_g$ might not be independent. But we are now supposing that the values $a_i$ and $b_i$ are uniform, thus $b_i-a_ig$ is also uniform due to Lemma~\ref{unifl}. Hence we have that $\mathrm{P}(b_i-a_ig\in\Sigma) =
|\Sigma|/q$ for any $i$ and any $g$, so
\[
\mathrm{P}(U_g) = \prod_{i=1}^M \frac{|\Sigma|}{q} = \left(\frac{|\Sigma|}{q}
\right)^M,
\]
and henceforth one obtains:
\[
\mathrm{P}(E|D=U) \leq q\left(\frac{|\Sigma|}{q}\right)^M.
\]
Setting all these values into Equation~\eqref{eq.1}, we eventually arrive at:
\[
\mathrm{P}(D=\mathcal{G}_\sigma | E) = 1 - \frac{\mathrm{P}(E|D=U)}
{\mathrm{P}(E|D=U)+1} \geq 1 - \mathrm{P}(E|D=U)\geq 1 -
q\left(\frac{|\Sigma|}{q} \right)^M.
\]
\item
Let $E^\prime$ stand for event ``Algorithm~\ref{alg1} returns \textbf{NOT PLWE}''. Again using Bayes' theorem we can write:
\begin{eqnarray*}
\mathrm{P}(D=U|E^\prime) & = & \frac{\mathrm{P}(E^\prime|D=U)\cdot\mathrm{P}(D=U)}
{\mathrm{P}(E^\prime|D=U)\cdot\mathrm{P}(D=U)+
\mathrm{P}(E^\prime|D=\mathcal{G}_\sigma)\cdot\mathrm{P}(D=\mathcal{G}_\sigma)}
\nonumber \\
& = & \frac{\mathrm{P}(E^\prime|D=U)}{\mathrm{P}(E^\prime|D=U) +
\mathrm{P}(E^\prime|D=\mathcal{G}_\sigma)}.
\end{eqnarray*}
But $\mathrm{P}(E^\prime|D=\mathcal{G}_\sigma)=1-\mathrm{P}(E|D=\mathcal{G}_\sigma)=0$, hence $\mathrm{P}(D=U|E^\prime)$, as required.
\end{enumerate}
\end{proof}

\begin{rem}We point out that the conclusion of \cite[Proposition~3.1]{ELOS:2015:PWI} is slightly different to ours. Namely, our conclusion when Algorithm~\ref{alg1} returns \textbf{NOT PLWE} is exactly the same, but in the case when Algorithm~\ref{alg1} returns \textbf{PLWE} or \textbf{NOT ENOUGH SAMPLES} our upper bound is slightly weaker. However, we think that in the proof provided in \cite[Proposition~3.1]{ELOS:2015:PWI} the probability $P(E|D=U)$ is underestimated. Moreover, in \cite{ELOS:2015:PWI} it is implicitly assumed that the sum of $N$ independent discrete truncated Gaussian distributions of variance $\sigma^2$ is a discrete truncated Gaussian distribution of variance $N\sigma^2$, as in the continuous non-truncated case. This claim needs further clarification not provided therein (see, for instance
\cite{GMPW:2020:IDG}). We have avoided this approach just considering that our smallness regions consist of all the possible values for the error polynomials evaluated at the roots.
\end{rem}

\begin{rem}
Notice that these attacks do not apply to cyclotomic moduli $\Phi_n(x)$. Indeed,
$\alpha=1$ is never a root modulo $q\nmid n$. Moreover, for $q\equiv 1\pmod{n}$
the order of each of the $\phi(n)$ different roots of $\Phi_n(x)$ is precisely
$n$.
\end{rem}

\subsection{Our method. Preliminary facts}

In this subsection, we present an attack against
PLWE for a ring $R_q:=\mathbb{F}_q[x]/(f(x))$ with $q$ prime and
$f(x)=\sum_{j=0}^Nf_jx^j\in\mathbb{Z}[x]$ monic and irreducible over
$\mathbb{Z}[x]$ of degree $N$. For the rest of this section, let us define
$N^*:=\left\lfloor\frac{N-1}{2}\right\rfloor$.

We start our study by assuming, first, that $x^2+1$ is irreducible over
$\mathbb{F}_q[x]$, and that $f(x)$ is divisible by $x^2+1$ in $\mathbb{F}_q[x]$.
Second, we will study the case where $f(x)$ has an irreducible quadratic factor
of the form $x^2+\rho$ with $\rho\in\mathbb{F}_q[x]$ with \emph{suitable}
multiplicative order modulo $q$.

Suppose first that $x^2+1$ is irreducible in $\mathbb{F}_q[x]$ and that
$x^2+1\mid f(x)$ in $\mathbb{F}_q[x]$. Hence, there exists
$\alpha\in\mathbb{F}_{q^2}\setminus\mathbb{F}_q$ with $\alpha^2=-1$ and
$f(\alpha)\equiv 0\pmod{q}$. Notice that $\Tr(\alpha)=0$, where $\Tr$ stands for
the trace of $\mathbb{F}_{q^2}$ over $\mathbb{F}_q$.

Now, if $(a(x),b(x)=a(x)s(x)+e(x))\in R_q^2$ is a PLWE sample attached to a
secret $s(x)$ and an error term $e(x)=\sum_{i=0}^{N-1}e_ix^i$, where the
coefficients $e_i$ are sampled from independent discrete Gaussians of mean $0$
and variance $\sigma^2$, then
\[
b(\alpha)-a(\alpha)s(\alpha)=e(\alpha),
\]
where now $s(\alpha)\in \mathbb{F}_{q^2}$, and
\begin{equation}
\Tr(b(\alpha)-a(\alpha)s(\alpha))=\Tr(e(\alpha))=\sum_{j=0}^{N-1}e_jt_j,
\label{partida}
\end{equation}
where $t_j=\Tr(\alpha^j)$. But if $j=2j'$, then
$t_j=\Tr((-1)^{j'})=2(-1)^{j'}$, otherwise $t_j=0$. Hence we are left with:
\begin{equation}
\Tr(b(\alpha)-a(\alpha)s(\alpha))=2\sum_{j=0}^{N^*}e_{2j}(-1)^j.
\label{traceerror}
\end{equation}
The sum of errors in Equation~\eqref{traceerror} takes values on the set $[-4N^*\sigma,4N^*\sigma]\cap\mathbb{Z}$.  Hence we can list the possible values $\frac{1}{2}\left(\Tr(b(\alpha)-a(\alpha)s(\alpha))\right)$. These elements are contained in the set 
\begin{equation}\label{sigma}
\Sigma:=[-2N^*\sigma, 2N^*\sigma]\cap\mathbb{Z},
\end{equation}
containing at most $4N^*\sigma+1$ elements, which we will term the \emph{smallness region} for this attack.

More in general, suppose now that for some non-zero $\rho\in\mathbb{F}_q$, the
polynomial $x^2+\rho$ is irreducible over $\mathbb{F}_q[x]$ and denote
by $r$ the multiplicative order of $-\rho$ modulo $q$. As before, assume that
$x^2+\rho\mid f(x)$ and that $\alpha^2=-\rho$  with $\alpha\in\mathbb{F}_{q^2}
\setminus\mathbb{F}_q$ so that $f(\alpha)\equiv 0\pmod{q}$. In this case, for a
PLWE-sample $(a(x),b(x)=a(x)s(x)+e(x))$ with $e(x)=\sum_{i=0}^{N-1}e_ix^i$, we
have:
\begin{equation}
\Tr(b(\alpha)-a(\alpha)s(\alpha))=\Tr(e(\alpha))=2\sum_{k=0}^{r-1}(-\rho)^k
\epsilon_k
\label{partida2}
\end{equation}
with $\epsilon_k := \sum_{i=0}^{\left\lfloor\frac{N^*}{r}\right\rfloor-1}e_
{2(ir+k)}$, for $0\leq k\leq r-1$.

Now each coefficient $\epsilon_k$ takes values on the set $-\left[2\frac{N^*}{r}\sigma, 2\frac{N^*}{r}\sigma\right]\cap\mathbb{Z}$. So, let us denote again by $\Sigma$ the set of all the possible
values of $\frac{1}{2}\Tr(e(\alpha))$. We have that 
\begin{equation}
|\Sigma|\leq \left(4\frac{N^*}{r}\sigma+1\right)^r.
\label{smregioncard}
\end{equation}
What makes \emph{suitable} the order of $-\rho$ to derive an attack is not
(just) being \emph{small}, but making the upper bound of Equation
(\ref{smregioncard}) much smaller than $q$, as we will analyse in the rest of
the Section.
We turn the previous ideas into an attack but first, for a reason which will be
transparent later, we define the following set:
\[
R_{q,0}=\{p(x)\in R_q: p(\alpha)\in\mathbb{F}_q\}.
\]
\begin{prop}The set $R_{q,0}$ is a ring, which, in particular, as an
$\mathbb{F}_q$-vector subspace of $R_q$ has dimension $N-1$.
\label{dim0}
\end{prop}
\begin{proof}It is obvious that $R_{q,0}$ is a ring and an $\mathbb{F}_q$-vector
subspace of $R_q$. As for the dimension, notice that for
$p(x)=\sum_{i=0}^{N-1}p_ix^i$, we have that $p(x)\in R_{q,0}$ if and only if
\[
p(\alpha)=\sum_{j=0}^{N^*}(-\rho)^jp_{2j}+\sum_{j=0}^{N^*}(-\rho)^jp_{2j+1}
\alpha\in\mathbb{F}_q,
\]
hence $p(\alpha)\in\mathbb{F}_q$ if and only if
$\sum_{j=0}^{N^*}(-\rho)^jp_{2j+1}=0$, hence the result follows.
\end{proof}

Let us set $S:=\mathbb{F}_{q^2}$, and assume that we are given a set of samples
from $R_{q,0}\times R_q$. The goal is to distinguish whether these samples come
from the PLWE distribution (but with the $a(x)$ term belonging to $R_{q,0}$) or
from a uniform distribution with values in $R_{q,0}\times R_q$.

We provide Algorithm~\ref{alg2} to achieve such goal that, in summary, runs as
follows: given a sample $(a_i(x),b_i(x))\in R_{q,0}\times R_q$, we pick a guess
$s\in S$ for $s(\alpha)$ and check whether $\frac{1}{2}\Tr(b_i(\alpha)-a_i
(\alpha)s)$ belongs to $\Sigma$. Since $a(\alpha)\in\mathbb{F}_q$, then
\[
\frac{1}{2}\Tr(b(\alpha)-a(\alpha)s)=\frac{1}{2}\Tr(b(\alpha))-\frac{1
}{2}a(\alpha)\Tr(s).
\]
Therefore, it is enough to check, for each $g\in\mathbb{F}_q$ (so that
$g$ is a putative value for $\Tr(s)$), whether or not
\[
\frac{1}{2}\Tr(b(\alpha))-\frac{1}{2}a(\alpha)g\in\Sigma.
\]
If this is not the case, we are actually removing from $S$ all the elements
$s'\in\mathbb{F}_{q^2}$ such that $\Tr(s')=g$. This means that we remove $q$
elements from $S$ at a time.

Remark that if the algorithm returns an element
$g\in\mathbb{F}_q$, we should understand that this element is just the trace of
one of the $q$ possible guesses for $s(\alpha)$. However, this is
(even if weaker than Algorithm~\ref{alg1}) enough as a decision attack.

\begin{figure}[ht]
\centering
\begin{tabular}[c]{ll}
\hline
\textbf{Input:} & A set of samples $C=\{(a_i(x),b_i(x))\}_{i=1}^M\in
R_{q,0}\times R_q$ \\
& A lookup table $\Sigma$ of all possible values for $\Tr(e(\alpha))$ \\
\textbf{Output:} & \textbf{PLWE},\\
                 & or \textbf{NOT PLWE},\\
                 & or \textbf{NOT ENOUGH SAMPLES}\\
\hline
\end{tabular}

\medskip

\begin{itemize}
\item \texttt{\emph{set}} $G:=\emptyset$
\item \texttt{\emph{for}} $g\in \mathbb{F}_q$ \texttt{\emph{do}}
	\begin{itemize}
	\item \texttt{\emph{for}} $(a_i(x),b_i(x))\in C$  \texttt{\emph{do}}
		\begin{itemize}
		\item \texttt{\emph{if}} $\frac{1}{2}\left(\Tr(b(\alpha))
                                -a(\alpha)g\right)\notin\Sigma$ \texttt{\emph{then}}
			\begin{itemize}
			\item \texttt{\emph{next}} $g$
			\end{itemize}
		\end{itemize}
	\item \texttt{\emph{set}} $G:=G\cup\{g\}$
	\end{itemize}
\item \texttt{\emph{if}} $G=\emptyset$ \texttt{\emph{then return}} \textbf{NOT
PLWE}
\item \texttt{\emph{if}} $|G|=1$ \texttt{\emph{then return}} \textbf{PLWE}
\item \texttt{\emph{if}} $|G|>1$ \texttt{\emph{then return}} \textbf{NOT ENOUGH
SAMPLES}
\end{itemize}
\hrule
\caption{Decision attack against $R_{q,0}$-PLWE}
\label{alg2}
\end{figure}

Actually, if the output of Algorithm~\ref{alg2} is such that $|G|=1$, say
$G=\{g\}$, Algorithm~\ref{alg2} (unlike Algorithm~\ref{alg1}) does not output a
guess for $s(\alpha)$; instead, $g$ is a guess for $\frac{1}{2}\Tr(s(\alpha))$
and in this case, we can venture that there exists $s(x)\in R_q$ such that
$b_i(x)-a_i(x)s(x)$ belongs to the error distribution for every sample
$(a_i(x), b_i(x))$.

Next, we evaluate the complexity of our attack in terms of
$\mathbb{F}_q$-multiplications:
\begin{prop}Given $M$ samples in $R_{q,0}\times R_q$, the number of
$\mathbb{F}_q$-multiplications required for Algorithm~\ref{alg2} is, at worst,
of order $\mathcal{O}(Mq)$.
\end{prop}
\begin{proof}
To begin with, given $g\in \mathbb{F}_q$:
\begin{itemize}
\item For each sample $(a_i(x),b_i(x))$, evaluating $\Tr(b_i(\alpha))$, requires
no multiplications in $\mathbb{F}_q$. Hence checking whether
$(\Tr(b(\alpha))-a(\alpha)g)/2\in\Sigma$ requires only $1$ multiplication in
$\mathbb{F}_q$.
\item In the worst case, the condition will fail for all the samples, in
which case we will perform $M$ multiplications in
$\mathbb{F}_q$ for each $g\in\mathbb{F}_q$.
\end{itemize}
Since the previous steps must be performed for every $g\in\mathbb{F}_q$, the
number of multiplications for the worst case will be $Mq$.
\end{proof}

Next, to obtain the probability of success of our attack, for
$p(x):=\sum_{j=0}^{N}p_{j}x^j$, let us denote
$p^*(x):=\sum_{j=0}^{N^*}p_{2j}x^j$. Denote, likewise
$R_q^*:=\mathbb{F}_q[x]/(f^*(x))$. Then, we have:
\begin{lem}Given a sample $((a(x),b(x))\in R_{q,0}\times R_q$ for
Algorithm~\ref{alg2} and given $g\in\mathbb{F}_q$,
$a(x)=\sum_{j=0}^{N-1}a_jx^j$ and $b(x)=\sum_{j=0}^{N-1}b_jx^j$,  then
\[
(\Tr(b(\alpha))-a(\alpha)g)/2\in\Sigma\mbox{ if and only if }b^*(-\rho)-ga^*(-\rho)\in \Sigma.
\]
In particular, the result of Algorithm~\ref{alg2} on samples
$(a_i(x),b_i(x))\in R_{q,0}\times R_q$ is exactly the result of Algorithm~\ref{alg1} applied
to the samples $(a_i^*(x),b_i^*(x))\in(R_q^*)^2$.
\label{remark1}
\end{lem}
\begin{proof}This is immediate from the fact that $f(\alpha)=f^*(-\rho)=0$.
\end{proof}

The following result will also be useful in our proof:
\begin{lem}Let $\{(a_i(x),b_i(x))\}_{i=1}^M$ be a set of input samples for
Algorithm~\ref{alg2}, where, as usual, the $a_i(x)$ are taken uniformly from
$R_{q,0}$ with probability $q^{1-N}$. Then, for the corresponding input samples
$(a_i^*(x),b_i^*(x))$ for Algorithm~\ref{alg1}, the elements
$a_i^*(x)$ are taken uniformly from $R_{q}^*$ with probability
$q^{1-N+N^\ast}=\mathcal{O}(q^{-N^\ast})$.
\label{unifprima}
\end{lem}
\begin{proof}
Let us write a sample, taken uniformly from $R_{q,0}$, as $a(x) = a_0(x) +
a_1(x)$, where
\[
a_0(x) = \sum_{j=0}^{N^*}a_{2j}x^{2j},\quad
a_1(x) = \sum_{j=0}^{N^*}\left(a_{2j+1}x^{2j}\right)x.
\]
For a start, observe that the polynomial $a(x)$ will be sampled from $R_{q,0}$
with probability $q^{1-N}$. But if we fix $a_0(x)$, any $N^\ast$-tuple $(a_1,a_3,
\dotsc,a_{2N^\ast+1})$ for $a_1(x)$ such that $\sum_{j=0}^{N^*}a_{2j+1}(-\rho)
^{2j}=0$ gives rise to the same polynomial $a^\ast(x)$. Since there are
$q^{N^*}$ of such tuples, and $a(x)$ was sampled uniformly, i.e., with
probability $q^{1-N}$, the sample $a^\ast(x)$ for Algorithm~\ref{alg1} will
occur with probability $q^{1-N}q^{N^\ast}$.
\end{proof}

We can now study the probability of success of our attack:

\begin{prop}Assume that $|\Sigma|<q$ and let $M$ be the number of input samples.
If Algorithm~\ref{alg2} returns \textbf{NOT
PLWE}, then the samples come from the uniform distribution on $R_{q,0}\times
R_q$. If it outputs anything other than \textbf{NOT PLWE}, then the samples are
valid PLWE samples with probability $1-q\left(|\Sigma|/q\right)^M$. In
particular, this probability tends to 1 as $M$ grows.
\label{pralg2}
\end{prop}
\begin{proof}
Let us denote by $X$ an unknown random distribution which can be either uniform or a PLWE one. The target of
Algorithm~\ref{alg2} is to determine which is the case. For a sample $v=(a(x), b(x))\in R_{q,0}\times R_q$ output by
$X$, we can write
$$
v=(a_0(x),b_0(x))+(a_1(x), b_1(x)),
$$
where $a_0(x)=a^*(x^2)$ and $b_0(x)=b^*(x^2)$, according to
Lemma~\ref{unifprima}, and $a_1(\alpha)=0$. Let $X^*$ be the random distribution
that, after sampling $X$ for a pair $(a(x),b(x))$, returns $(a^*(x),b^*(x))\in
R_q^*\times R_q^*$.

Notice that if $X$ is a uniform distribution, then Lemma~\ref{unifprima}
ensures that $X^*$ is also a uniform distribution with values on $R_q^*\times
R_q^*$.
As for the reciprocal, suppose that $X^*$ is uniform, but $X$ is PLWE with a
truncated Gaussian distribution associated to the error.
Then consider pairs $(a_i^\ast(x), b_i^\ast(x))$ output by $X^\ast$; since by
hypothesis $X^\ast$ is uniform, evaluating $(a_i^*(x),b_i^*(x))$ at $-\rho$ should
yield a uniform distribution of values, by virtue of Lemma~\ref{unifprima} and
Lemma~\ref{unifl}; however, that evaluation equals $(a_i^*(-\rho),a_i^*
(-\rho)s^*(-\rho)+e_i^*(-\rho))$, where each $a_i^*(-\rho)$ is indeed uniformly
sampled due to Proposition~\ref{unifprod} and Lemma~\ref{unifl}; but each
$e_i^*(-\rho)$ is the output of the truncated Gaussian associated to $X$, which we
assumed PLWE.  Thus, $X^\ast$ turns out to be not uniform, against the hypothesis.
In summary, the distribution $X$ is uniform if and only if $X^*$ is uniform.

If $X$ is the PLWE distribution, for a sample $(a(x), b(x))$ it will be
$b(x)=a(x)s(x)+e(x)$ with $e(x)$ drawn from an $R_q$-valued discrete truncated Gaussian,
hence setting $a(x)=a_0(x)+a_1(x)$, and likewise for $b(x)$, $s(x)$ and $e(x)$
we have that the even-degree part of $b(x)$ is
$$
b_0(x)=a_0(x)s_0(x)+a_1(x)s_1(x)+e_0(x).
$$
Notice that within Algorithm~\ref{alg2}, testing the sample $(a(x),b(x))$ in
each iteration will return the same value as if the test were applied to the
sample $(a_0(x),a_0(x)s_0(x)+e_0(x)) = (a^\ast(x^2),b^\ast(x^2))$, since
$a_1(\alpha)=0$. But then Lemma~\ref{unifprima} ensures that
Algorithm~\ref{alg2} gives the same result on sample $(a(x), b(x))$ as
Algorithm~\ref{alg1} on sample $(a^\ast(x^2), b^\ast(x^2))$.

Now, set input samples $S=\{(a_i(x),b_i(x))\}_{i=1}^M$ for Algorithm~\ref{alg2}
produced by $X$ and let $S^*=\{(a_i^*(x), b_i^*(x))\}_{i=1}^M$ be the
corresponding ``starred'' samples.

Let us define the following events:
\begin{itemize}
\item $E_q=$ The random distribution $X$ is uniform.
\item $E_q^*=$ The random variable $X^*$ is uniform.
\item $rP=$ \emph{Algorithm~\ref{alg2} returns \textbf{PLWE}} on input
samples $S$.
\item $rP^*=$ \emph{Algorithm~\ref{alg1} returns \textbf{PLWE}} on input
samples $S^*$.
\item $rNE=$ \emph{Algorithm~\ref{alg2} returns \textbf{NOT ENOUGH SAMPLES}}
on input samples $S$.
\item $rNE^*=$ \emph{Algorithm~\ref{alg1} returns \textbf{NOT ENOUGH
SAMPLES}} on input samples $S^*$.
\item $rNP=$ \emph{Algorithm~\ref{alg2} returns \textbf{NOT PLWE}} on input
samples $S$.
\item $rNP^*=$ \emph{Algorithm~\ref{alg1} returns \textbf{NOT PLWE}} on
input samples $S^*$.
\end{itemize}

As mentioned in Lemma~\ref{remark1}, the outcome of Algorithm~\ref{alg2} on
input samples $S$ is the same as the outcome of Algorithm~\ref{alg1} on input
samples $S^*$, hence
\[
rP\cup rNE=rP^*\cup rNE^*\mbox{ and }rNP=rNP^*.
\]

Moreover, as stated before, we have that $E_q=E_q^*$, so that
\[
P[E_q|(rP\cup rNE)]=P[E_q|(rP^*\cup rNE^*)]= P[E_q^*|(rP^*\cup rNE^*)]
\]
and in this case the result follows from Proposition~\ref{pralg1}.

On the other hand, since $rNP=rNP^*$, if Algorithm~\ref{alg2} returns
\textbf{NOT PLWE} (which happens precisely if Algorithm~\ref{alg1} returns
\textbf{NOT PLWE}), by Proposition~\ref{pralg1}, the distribution $X^*$ will
be uniform and hence the distribution $X$ will be uniform as well.

\end{proof}

\subsection{The attack. Stage 2}

Suppose we are given access to an oracle that generates samples
$\{(a_i(x),b_i(x)\}_{i\geq 1}\subseteq R_q^2$ either from the uniform
distribution or from the PLWE distribution. We claim that we can decide, with
non-negligible advantage, and in probabilistic polynomial time depending on $N$,
which distribution they come from.

We need some tools to justify our claim.
First of all, we define Algorithm~\ref{alg4}, which can be also thought of as
random variable $X_0$, with values on $R_{q,0}\times R_q$.

\begin{figure}[ht]
\centering
\begin{tabular}[c]{ll}
\hline
\textbf{Input:} & A distribution $X$ over $R_q^2$ \\
\textbf{Output:} & A sample $(a(x),b(x))\in R_{q,0}\times R_q$ \\
\hline
\end{tabular}

\medskip

\begin{itemize}
\item \texttt{\emph{set}} $count:=0$
\item \texttt{\emph{do}}
	\begin{itemize}
	\item	$(a(x),b(x)) \overset{X}{\leftarrow} R_q^2$
	\item $count \leftarrow count+1$
	\end{itemize}
\item \texttt{\emph{until}} $a(x)\in R_{q,0}$
\item[]
\item \texttt{\emph{return}} $(a(x),b(x)), count$
\end{itemize}
\hrule
\caption{Random variable $X_0$}
\label{alg4}
\end{figure}

In a nutshell, Algorithm~\ref{alg4} repeatedly samples the input variable $X$
until it obtains a pair $(a(x),b(x))$ such that $a(x)\in R_{q,0}$ and returns it
along with the number of invocations to $X$. Regarding this random variable
$X_0$, we have the following

\begin{prop}
Notations as above,
\begin{itemize}
\item[\emph{a)}] If $X$ is uniform over $R_q^2$ then $X_0$ is uniform over
$R_{q,0}\times R_q$.
\item[\emph{b)}] If $X$ is a PLWE distribution then $X_0$ is a $R_{q,0}\times R_q$-valued
PLWE distribution.
\end{itemize}
\end{prop}
\begin{proof}
For each $i\geq 1$, let $B_i$ denote the event that $X^{(i)}\notin R_{q,0} \times
R_q$, where $X^{(i)}$ is the outcome of the $i$-th running of the variable $X$.
Notice that by independence, the probabilities $P[B_i]$ are equal for any $i$.
Likewise, for each $v=(a(x),b(x))\in R_q^2$, denote by $C_{i,v}$ the event that
$X^{(i)}=v$. Again, by independence, the probabilities $P[C_{i,v}]$ are equal to
$P[X=v]$ for any $i$. We can write:
\begin{equation}
P[X_0=v]=P[C_{1,v}]+\sum_{i=2}^{\infty}P[B_1 \cap\cdots\cap B_{i-1}\cap C_{i,v}].
\label{finaleq}
\end{equation}
Taking into account that $P[X^{(i)}\in R_{q,0}\times R_q]=q^{-1}$, the right hand
side of Equation~\eqref{finaleq} equals
\begin{equation}
P[X_0=v]=P[X=v]\sum_{i=0}^{\infty}(1-q^{-1})^{i} = qP[X=v].
\label{auxfinaleq}
\end{equation}
Hence, if $X$ is uniform, then for any $i\geq 1$ we have that $P[X=v]=q^{-2N}$
and hence Equation~\ref{auxfinaleq} equals $q^{1-2N}$, thus $X_0$ is uniform.

Assume now that $X$ is a PLWE distribution. Any $v=(a(x),b(x))$ sampled from
$X$ has $a(x)$ sampled from an $R_q$-valued uniform distribution and
$b(x)=a(x)s(x)+e(x)$, where $e(x)$ is sampled from $\chi_{R_q}$, an $R_q$-valued
discrete truncated Gaussian variable. Taking into account that for any $i\geq 1$, it
holds that
\[
P[C_{i,v}]=P[X=v]=P[U=a(x)\cap \chi_{R_q}=e(x)],
\]
Equation~\eqref{finaleq} reads now
\begin{multline}
P[X_0=v] = qP[X=v] = qP[U=a(x)]P[\chi_{R_q}=e(x)] \\
= q^{1-N}P[\chi_{R_q}=e(x)],
\label{eqfinal1}
\end{multline}
so that actually $X_0$ is a PLWE distribution with values on $R_{q,0}\times
R_q$.
\end{proof}

\begin{rem}
Within $X_0$, the variable $count$ stores the expected number of times one has
to run $X$ before succeeding. Since this random variable is distributed as a
geometric random variable of first kind with success probability $q^{-1}$, the
expected number of times is just its mean $E[count]=q=\mathcal{O}(N^2)$ (see
for instance \cite[Appendix A]{BP:2015:MAI}).
\end{rem}

We are now almost ready to present a general attack, where an $R_q$-valued random
distribution $X$ (an oracle), which can be either uniform or PLWE, is provided
to the attacker. The goal of this attacker is to distinguish which one out of the
two it actually is. Before proceeding, we need the help of a pair of parameters.

First, we set an upper bound $l=\mathcal{O}(p(N))$ for the maximal
number of samples to generate from our oracle, where $p(x)$ is a polynomial of
suitable small degree and we keep only those $(a_i(x),b_i(x))\in R_{q,0}\times
R_q$ from the samples we generate.

Notice that for a given sample in $R_q^2$, the probability that it actually
belongs to $R_{q,0}\times R_q$ is $1/q$, hence, the probability that for
$l$ samples in $R_q^2$, at least $k\leq l$ of them belong to $R_{q,0}\times R_q$
can be approximated (if $l/q>5$) as
\begin{equation}
P[B(l,q^{-1})\geq k]\cong
P\left[
\mathcal{N}\left(lq^{-1},\sqrt{lq^{-1}(1-q^{-1})}\right)\geq k
\right],
\label{prob1}
\end{equation}
where $B(l,q^{-1})$ stands for a binomial distribution of parameters $l$ and
$q^{-1}$ and $\mathcal{N}\left(lq^{-1},\sqrt{ lq^{-1}(1-q^{-1})}\right)$ is a
Gaussian distribution of mean $lq^{-1}$ and standard deviation
$\sqrt{lq^{-1}(1-q^{-1})}$. But the right hand side of Equation~\eqref{prob1} can
be written as
\[
P\left[Z\geq \frac{k-lq^{-1}}{\sqrt{lq^{-1}(1-q^{-1})}}\right],
\]
where $Z$ is a standard Gaussian distribution. Hence, the probability that, out
from $l$ samples in $R_q^2$ we obtain $k$ of them from $R_{q,0}\times R_q$ is
about $1/2$ if $k\approx lq^{-1}>5$.

Next, we must choose $k$ so that $1-q\left(|\Sigma|/q\right)^k\geq
\theta$, where $\theta$ is our desired success probability threshold.

We present next Algorithm~\ref{alg3}, which is our main result: a general
decision attack over an unknown $R_q$-valued random distribution $X$ which,
as stated above, can be either uniform or PLWE. Remark that
Algorithm~\ref{alg3} uses Algorithm~\ref{alg4} and Algorithm~\ref{alg2}
as subroutines.

\begin{figure}[ht]
\centering
\begin{tabular}[c]{ll}
\hline
\textbf{Input:} & A distribution $X$ over $R_q^2$ which can be either uniform or
PLWE \\
                & A pair of values $(k,l)\in\mathbb{N}$ \\
\textbf{Output:} & \textbf{UNIFORM}, \\
              & or \textbf{PLWE}, \\
              & or \textbf{NOT ENOUGH SAMPLES}, \\
		 & or \textbf{FAIL} \\

\hline
\end{tabular}

\medskip

\begin{itemize}
\item \texttt{\emph{set}} $C:=\emptyset$, $count:=0$, $c:=0$
\item \texttt{\emph{do}}
	\begin{itemize}
	\item \texttt{\emph{set}} $(a(x),b(x)), c :=$ Algorithm.\ref{alg4}$(X, R_q)$
	\item $C \leftarrow C \cup\{(a(x),b(x))\}$
	\item $count \leftarrow count+c$
	\end{itemize}
\item \texttt{\emph{until }} $|C| = k$ \texttt{\emph{or}} $count>l$
\item \texttt{\emph{if }} $count\leq l$ \texttt{\emph{ then}}
		\begin{itemize}
			\item \texttt{\emph{set}} $\omega :=$ Algorithm.\ref{alg2}$(C)$
			\item \texttt{\emph{return}} $\omega$
		\end{itemize}
\item \texttt{\emph{ else}}
		\begin{itemize}
			\item \texttt{\emph{return }\textbf{FAIL}}
		\end{itemize}
\end{itemize}
\hrule
\caption{Decision attack against $R_q$-PLWE}
\label{alg3}
\end{figure}

\begin{rem}
Some words on the feasibility of our attack for explicit parameters are in
order. First, when the modulus polynomial is $\Phi_n(x)$, the $n$-th cyclotomic
polynomial, it is well known that for a prime $q$, letting $f$ be the smallest
natural number such that $q^f\equiv 1\pmod{n}$, we have a factorisation of
$\Phi_n(x)$ into $\phi(n)/f$ irreducible factors over $\mathbb{F}_q$ of degree
$f$. In particular, if $f=1$, then $\Phi_n(x)$ is totally split over
$\mathbb{F}_q$ and its roots have maximal order $n$.

More in general (see \cite{WZFY:2017:EFC} pag. 202 for details) if $n=p^k$, and
$q=1+p^Au$ with $u$ coprime to $p$ and $2\leq A<k$, denoting by $\Omega(p^A)$
the group of $p^A$-th primitive roots of unity in $\mathbb{F}_q$, we have
\begin{equation}
\Phi_{p^k}(x)=\prod_{\rho\in\Omega(p^A)}(x^{p^{k-A}}-\rho).
\label{chinese}
\end{equation}
In this case, for minimal possible order $p^2$, the attack against
$R_{q,0}\times R_q$ is extremely efficient and a small number of samples (about
6) are enough to solve the decisional version of the problem with overwhelming
probability. However, the subspace $R_{q,0}$ has a dimension about $N/2$ and the
probability that a sample $(a(x),b(x))\in R_q^2$ actually lies in $R_{q,0}\times
R_q$ is of the order of $q^{-N/2}$ making the samples for Algorithm~\ref{alg2}
extremely unlikely to occur.

However, the sufficient condition for Algorithm~\ref{alg2} to work is that
$|\Sigma|<q$ and since the look-up table must be only constructed once, for
some parameter instances where building this table is still feasible, our attack
will succeed, as proved in Proposition~\ref{pralg2}. 
\end{rem}

\begin{ex}
Consider the polynomial
$$
f(x)=x^{256}+92882x^{254}+x^{128}+92882x^{126}+x^2+92882,
$$ and $q=100003$. The irreducible factor
$x^2+92882$ divides $f(x)$, $r:=\mathrm{ord}_q(-92882)=6$. For $\sigma=0.04$
we have
\[
|\Sigma|\leq \left( 4\frac{N^*}{r}\sigma+1\right)
^r<7126<q.
\]
Hence, $\frac{|\Sigma|}{q}<0.07126$ and Proposition~\ref{pralg2} predicts that
the attack on $R_{q,0}\times R_q$ will succeed at least $81.62\%$ of times with
just $5$ samples and $99.993\%$ with $8$ samples. This example works with a root of significantly larger order as those considered in \cite{ELOS:2015:PWI}.
\end{ex}

\begin{ex}
Consider the polynomial
$$
f(x)=x^{80}+x^{78}-x^{66}-x^{64}-2x^{48}-2x^{46}+2x^{34}+2x^{32}+x^{16}+x^{14}-x^2+262
$$ 
and $q=263$. The irreducible factor $x^2+1$ divides $f(x)$ and $r:=\mathrm{ord}_q(-1)=1$. For $\sigma=0.18$
we have
\[
|\Sigma|\leq \left( 4\frac{N^*}{r}\sigma+1\right)
^r<227<263<q.
\]
Hence, $\frac{|\Sigma|}{q}<0.863$ and Proposition~\ref{pralg2} predicts that
the attack on $R_{q,0}\times R_q$ will succeed at least $96\%$ of times with 
just $60$ samples. Notice that $q=263$ is close to byte size (251, 253). These byte-size primes are used in the LAC cryptosystem. We have chosen $N=80$ to simplify the computations, while LAC requires at least $N=128$. Moreover, in LAC, the modulus polynomial is cyclotomic. However, this example shows that our attack applies to instances where the parameters are relatively close to those used in practical scenarios.
\end{ex}

\begin{ex}
Consider the polynomial
$$
f(x)=x^{32}+x^{30}+x^2+264,
$$ 
and $q=263$. The irreducible factor $x^2+1$ divides $f(x)$ and $r:=\mathrm{ord}_q(-1)=1$. For $\sigma=0.502$
we have
\[
|\Sigma|\leq \left( 4\frac{N^*}{r}\sigma+1\right)
^r<261<263<q.
\]
Hence, $\frac{|\Sigma|}{q}<0.9933$ and Proposition~\ref{pralg2} predicts that
the attack on $R_{q,0}\times R_q$ will succeed at least $87\%$ of times with  $1000$ samples. Notice that in this example, the variance can be chosen to be considerably larger than in previous examples, also along the lines of practical scenarios. 
\end{ex}

\begin{rem}
Again, in the cyclotomic case, it is very unlikely for a random couple $(n,q)$
with $q$ prime, that $\Phi_n(x)$ has an irreducible quadratic factor with zero
linear term. First, it must be $q\equiv -1\pmod{n}$ (see, for instance,
\cite{WZFY:2017:EFC} page 203 for $n=3^k$). Certainly, if $\Phi_n(x)$ factors as
a product of irreducible quadratic factors with non-zero linear term, a change
of variable $y:=x-\rho$ for suitable $\rho\in\mathbb{F}_q$ grants a quadratic
factor of the required shape of the (non-cyclotomic) polynomial
$\Phi_n(y-\rho)$. However, all our computations show that the order of $\rho$ is
$\phi(n)$ in nearly all the cases.
\end{rem}

These examples show that our attack applies preferably to
non-cyclotomic polynomials, although some caution must be adopted even in the
cyclotomic case. It would be very interesting to have a more detailed study,
specially for composite cyclotomic conductors. Beyond this, in the next section
we display several instances of our attack against non-cyclotomic parameters,
which is successful with very large probability and with a very small number of
samples from the $R_q^2$-distribution.

\section{Coding examples}
We have developed a Maple program in order to check the real applicability and
performance of our Algorithms. The code can be found at
\textsc{GitHub}%
\footnote{\url{https://github.com/raul-duran-diaz/PLWE-ZeroTraceAttack}}.
We have tested the code in Maple version 10.00 (a pretty old version), so it
should
be also executable over virtually any other higher version.
As a disclaimer, keep in mind that we made no attempt at optimising the code:
our aim was just
presenting a ``proof of concept'' of the validity and feasibility of the
algorithms contributed in this work. Depending on the available hardware and
software, the execution may appear a little bit ``sluggish''.

The code implements a number of procedures to carry out the different needed
tasks. Maple has an internal uniform sampler (called \texttt{rand}) that can be
easily adjusted to give samples uniformly over any finite set of numbers. In
particular, we use it over the base field $\mathbb{F}_q$.
We implement an approximation of discrete Gaussian based on the regular normal
Gaussian distribution provided directly by Maple.

Regarding the provided Maple sheet, we provide in particular the implementation
of Algorithm~\ref{alg3} (case $\rho=1$). To run any example, one just selects
the desired parameters (in the so-called ``main section'') and executes the
sheet right from the beginning to the end. We have retained (roughly) the same
notation for the Maple sheet as the one used along the present work, to ease the
use of the sheet.

The execution comprises the following steps:
\begin{enumerate}
\item Generating samplers for the uniform distribution (\texttt{rollq}) and for
the discrete Gaussian (\texttt{X}).
\item Obtaining a prime of the desired size such that $-1$ is not a quadratic
residue $\mathbb{F}_q$.
\item Computing an irreducible polynomial in $\mathbb{Z}[x]$ such that
it is divisible by $x^2+1$ in $\mathbb{F}_q[x]$.
\item Computing the set $\Sigma$ as a function of the parameters, following
Equation~\eqref{sigma}.
\item Selecting a number of executions (variable \texttt{ntests}) for
Algorithm~\ref{alg2}, and a number of samples for each execution (variable
\texttt{M}). Then
\begin{enumerate}
\item First, a loop is executed \texttt{ntests} times and, for each turn,
\texttt{M}
samples from the PLWE distribution are generated, and passed to Algorithm~\ref{alg2}.
If it returns anything different from a set containing only one element, the
execution is recorded as a failure.
\item Second, another loop is executed, but generating the samples from the
uniform distribution, and passing each set of samples to Algorithm~\ref{alg2}. If it
returns anything other than an empty set, the execution is recorded as a
failure.
\end{enumerate}
\end{enumerate}

We tested the implementation using two examples with different parameters that
we summarise in the following table:

\medskip

\begin{center}
\begin{tabular}[c]{lcc}
\hline
\textbf{Parameter} & \textbf{Example 1} & \textbf{Example 2} \\
\hline
$q$                           & $5003$  & $10007$ \\
$N$                           & $62$    & $102$ \\
$\sigma$                      & $8$     & $8$ \\
\texttt{ntests}               & $20$    & $20$ \\
\texttt{M}                    & $5$     & $5$ \\
Average sampling exec. time   & $\simeq 4.97$   & $\simeq 16.65$ \\
Average test exec. time       & $\simeq 0.94$   & $\simeq 3.02$ \\
\hline
\end{tabular}
\end{center}

\medskip

The executions of the two examples gave no failures for the uniform distribution,
but fails approximately $5\%$ to $10\%$ of the tests for the PLWE distribution. To
understand why, keep in mind that Algorithm~\ref{alg2} is essentially
probabilistic and requires that there exists a $g\in\mathbb{F}_q$ such that
$(\Tr(b(\alpha))-a(\alpha)g)/2 \in\Sigma$ for all samples in the set of samples
$C$.

Regarding execution times (given in seconds), the examples have been executed
and measured on our
platform, a Virtual Box configured with 1 GB of main memory running over Intel
CORE i5, at 3.4 GHz, which is rather slow. The execution time for the test and
for creating samples from both distributions, namely uniform and PLWE, are very
similar; therefore we just give one figure for the execution time and one
figure for the sampling time, which is applicable to both distributions. It is
also apparent that sampling is by far the most time-consuming operation.

\end{document}